\documentclass[aps,prl,reprint,10pt,superscriptaddress,nofootinbib]{revtex4-1}

\usepackage{amsmath,amssymb,amsfonts}
\usepackage{bm}
\usepackage[colorlinks=true,linkcolor=blue,citecolor=blue,urlcolor=blue]{hyperref}
\newcommand{\doilink}[2]{\href{https://doi.org/#1}{#2}}
\newcommand{\arxivlink}[2]{\href{https://arxiv.org/abs/#1}{#2}}
\newcommand{\ket}[1]{|#1\rangle}
\newcommand{\bra}[1]{\langle #1|}
\newcommand{\ketbra}[1]{\ket{#1}\bra{#1}}
\newcommand{\Ker}{\operatorname{Ker}}
\newcommand{\Ran}{\operatorname{Ran}}
\newcommand{\rank}{\operatorname{rank}}
\newcommand{\supp}{\operatorname{supp}}
\newcommand{\Tr}{\operatorname{Tr}}
\newcommand{\Span}{\operatorname{span}}
\newtheorem{theorem}{Theorem}
\newtheorem{lemma}{Lemma}

\begin{document}

\title{Every Rank-2 Bipartite Entangled State is Projectively Steerable}

\author{Yu-Xuan Zhang}
\affiliation{School of Physics, Nankai University, Tianjin 300071, People's Republic of China}

\author{Jing-Ling Chen}
\email{chenjl@nankai.edu.cn}
\affiliation{Theoretical Physics Division, Chern Institute of Mathematics, Nankai University, Tianjin 300071, People's Republic of China}

\date{\today}

\begin{abstract}
We prove that every rank-2 bipartite entangled state is already steerable with projective measurements, closing the first possible gap between entanglement and Einstein-Podolsky-Rosen (EPR) steering. In arbitrary finite local dimensions, any rank-2 entangled state is projectively steerable in at least one direction; when the effective dimensions are equal, steering is two-way. The argument is geometric: a rank obstruction forces some projective outcome on the larger local system to reach the boundary of the trusted-state space. From this boundary contact, a neighboring measurement produces a linear displacement that cannot be reproduced by a local-hidden-state model, because positivity only permits a quadratic filling into a kernel direction. This same block witnesses partial-transpose entanglement. If the contact is degenerate, a Schur complement removes one product layer and reduces the problem to a lower-rank entangled residual; each such reduction lowers the rank, so the recursion terminates. More broadly, we derive a directional rank criterion: if an $m\otimes n$ entangled state $\rho$ satisfies $\rank\rho\le 1+\lfloor(m-1)/(n-1)\rfloor$, then it is projectively steerable from $A$ to $B$. Since counterexamples already occur at rank three, rank two is the unique mixed rank for which entanglement guarantees projective steering, and the proof provides an operational certificate based solely on the support and kernel of the density matrix.
\end{abstract}

\maketitle

\textit{Introduction.---}The EPR paper exposed the possibility of remotely assigning quantum states, and Schr\"odinger named this asymmetric phenomenon ``steering''~\cite{Einstein1935,Schrodinger1935}.  Bell nonlocality later separated classical locality from quantum correlations~\cite{Bell1964}, while the modern formulation of Wiseman, Jones, and Doherty made steering an operational resource: an untrusted party steers a trusted party precisely when the observed assemblage has no local-hidden-state (LHS) model~\cite{Wiseman2007,Jones2007}.  The hierarchy is strict: Werner-type constructions and one-way steering examples show that entanglement alone is not a steering certificate~\cite{Werner1989,Bowles2014,Uola2020,Cavalcanti2017}.  Operationally, steering also connects to semidefinite programs, steering monotones, and joint-measurability formulations~\cite{Skrzypczyk2014,Quintino2014,UolaMoroder2014,UolaBudroni2015}.  Pure entangled states, however, already steer under suitable projective measurements within their Schmidt supports.  The first place where a genuine mixed-state separation could appear is therefore rank two.  The sharp structural question is not whether all entangled states steer, but whether the hierarchy already splits at rank two or the entanglement--steering gap remains closed at this first mixed rank.

This work settles the question at the first genuinely mixed rank. Rank one is the pure-state boundary; rank two is the first essential mixed-state test of the hierarchy.  We prove that every rank-2 bipartite entangled state, in arbitrary finite local dimensions after removing null local supports, is already steerable with projective measurements in at least one direction.  If the two effective local dimensions are equal, the conclusion is two-way.  The corresponding two-qubit result was established previously~\cite{ZhangChenFundamental2025}.  Thus there is no exceptional rank-2 entangled family hidden behind a projective LHS model, while analytic results for Bell-diagonal states already exhibit entangled rank-three states without projective steering~\cite{Nguyen2019}.  The proof works directly in the full effective local dimensions: the support--kernel geometry of the density operator furnishes the steering certificate as a structural alternative to two-qubit reduction and steering-inequality optimization.

The local boundary criterion, proved in the Supplemental Material~\cite{SM} and developed more broadly in Ref.~\cite{ZhangChenBoundary2026}, has a simple physical content.  It is complementary to geometric pictures of steering based on steering ellipsoids and EPR-correlation bodies~\cite{Jevtic2014,Jevtic2015,Nguyen2019}.  When a trusted conditional state lies on a boundary face, positivity fills a kernel direction only to second order.  If a neighboring projective outcome produces a first-order support--kernel tangent block, no LHS model can reproduce the resulting linear-versus-quadratic scaling.  The same matrix element produces a negative partial-transpose minor, in line with the close relation between negativity and steering~\cite{Pusey2013}.  Thus one boundary contact converts a local coherence into both projective steering and partial-transpose negativity.

The global step is a rank obstruction.  A rank-$r$ state on $\mathbb C^m\otimes\mathbb C^n$ confines the outcome-dependent Bob contractions to an $m$-dimensional subspace of $M_{n\times r}(\mathbb C)$.  If this space is too large to avoid the rank-one determinantal variety, then some projective outcome on the $m$-dimensional party prepares a pure conditional state on the $n$-dimensional party.  Rank two is the first mixed case of this mechanism, but the Schur-complement branch is in fact recursive.  If a contact is locally degenerate, one product layer can be peeled off; the entangled residual has rank and effective untrusted dimension lowered by one, so the same rank-forcing inequality remains valid.  The same mechanism yields a broader directional rank hierarchy for steering: high local dimension on the measuring side forces many mixed ranks to be steerable, while the reverse orientation may leave no nontrivial mixed-rank window.

This formulation ties the EPR steering hierarchy to low-rank PPT separability, determinantal algebraic geometry, and support--kernel certificates.  Its practical significance is that the support and kernel of a low-rank state are often the most accessible and stable data: they can come directly from source constraints, state reconstruction, variational or tensor-network descriptions, dissipative dark-state manifolds, or a chosen bipartite cut of a larger system.  Such structural information is available before, and often more robust than, an optimized steering inequality or a full SDP search.  The theorem turns rank/support data from a reconstruction or model into a steering decision throughout rank two; beyond rank two, the same support--kernel scalar remains a directly checkable certificate.


We use the following terminology. For concreteness, we refer to the untrusted party as Alice and the trusted party as Bob. Projective measurements on Alice are taken to mean a complete set of rank-one projective measurements; their outcomes generate the unnormalized projective assemblage. A state is \emph{projectively steerable} from $A$ to $B$ if this assemblage has no local-hidden-state representation.

Throughout, $\ket{x}\bra{x}$ denotes the unnormalized rank-one operator associated with $x$. If $S$ is a subspace, $P_S$ denotes the orthogonal projection onto $S$; for a nonzero vector $x$, we write $P_x:=P_{\mathbb Cx}$. Vectors that label actual projective outcomes are taken to be unit representatives.

Let $\rho$ be a density operator on $\mathbb C^m\otimes\mathbb C^n$ with rank $r$, and consider the steering direction $A\to B$. We take $m$ and $n$ to be the effective local support dimensions of $\rho$. Thus no nonzero vector $\alpha\in\mathbb C^m$ satisfies $(\bra{\alpha}\otimes I)\psi=0$ for every $\psi\in\Ran\rho$.

If, for example, $\rho_A$ has a null vector, we restrict to $\supp\rho_A$, since that null local direction never appears as a projective outcome with nonzero probability. After this reduction, the contraction map used in the next section is injective, so the dimension of the family of conditional vectors is exactly the effective dimension of the untrusted side.

The main result is the following.
\begin{theorem}[Rank-Two Projective Steering]
\label{thm:ranktwo}
Let $\rho$ be a rank-2 entangled state with effective local dimensions $m$ and $n$. If $m\ge n$, then $\rho$ is projectively steerable from $A$ to $B$. If $n\ge m$, then $\rho$ is projectively steerable from $B$ to $A$. For $m=n$, both directions hold.
\end{theorem}
The proof passes through a broader low-rank statement that locates the boundary contact underlying the rank-two result.

For the proof, a \emph{pure contact} for the direction $A\to B$ is a projective outcome whose unnormalized Bob state is nonzero and rank one.

\textit{Rank-forced pure contacts.--} Choose a square-root decomposition $\rho=\sum_{j=1}^r \ket{\psi_j}\bra{\psi_j}$, where the $\psi_j$ are linearly independent but need not be normalized. For a covector $x\in(\mathbb C^m)^*$, set $V(x)=[(x\otimes I)\psi_1,\ldots,(x\otimes I)\psi_r]$, an $n\times r$ matrix; for a unit representative of a projective outcome $\alpha$ we use $x=\bra{\alpha}$. Then the unnormalized Bob state is $\sigma_\alpha=\Tr_A[(P_{\alpha}\otimes I)\rho]=V(\bra{\alpha})V(\bra{\alpha})^\dagger$. Hence $\sigma_\alpha$ is a nonzero pure conditional state exactly when $V(\bra{\alpha})$ has rank one. Equivalently, for some line $\mathbb C\phi\subset\mathbb C^n$,
\begin{equation}
    \ket{\alpha}\otimes \phi^\perp\subseteq\Ker\rho.
    \label{eq:tensorzero}
\end{equation}
Here the inclusion means that $\ket{\alpha}\otimes\ket{\beta}\in\Ker\rho$ for every $\beta\perp\phi$. The covector map $x\mapsto V(x)$ is linear; the passage from a ket $\alpha$ to $\bra{\alpha}$ is only the usual antilinear Riesz identification. The vectors orthogonal to $\phi$ are precisely the zero directions of $\sigma_\alpha$, and positivity converts zero expectation values into kernel vectors.

On the effective Alice support this covector map is injective. Its image is therefore an $m$-dimensional linear subspace $L\subset M_{n\times r}(\mathbb C)$. In projective space $\mathbb P(M_{n\times r})$, the nonzero rank-one matrices form the Segre variety $\Sigma_1=\mathbb P^{n-1}\times\mathbb P^{r-1}$ of dimension $n+r-2$. The projectivization $\mathbb P L$ has dimension $m-1$, while the ambient projective dimension is $nr-1$. Over $\mathbb C$, the projective dimension theorem for projective varieties~\cite{Harris1992} implies that avoiding $\Sigma_1$ would require $(m-1)+(n+r-2)<nr-1$, equivalently $m\le (n-1)(r-1)$. Therefore the strict reverse inequality, $m>(n-1)(r-1)$, forces a nonzero rank-one matrix in $L$. Equivalently,
\begin{equation}
    r\le 1+\left\lfloor \frac{m-1}{n-1}\right\rfloor,
    \label{eq:rankbound}
\end{equation}
forces the tensor-zero slice in Eq.~\eqref{eq:tensorzero} for the direction $A\to B$. Exchanging $m$ and $n$ gives the corresponding condition for $B\to A$.

Equivalently, we have proved the following auxiliary lemma.
\begin{lemma}[Rank-forced pure contact]
\label{lem:rankforced}
Let $\rho$ have effective local dimensions $m,n$ and rank $r$. If
\begin{equation}
    r\le 1+\left\lfloor \frac{m-1}{n-1}\right\rfloor,
\end{equation}
then $\rho$ has a pure contact for $A\to B$. Equivalently, for some projective outcome $\alpha$ and Bob line $\mathbb C\phi$, one has $\ket{\alpha}\otimes\phi^\perp\subseteq\Ker\rho$.
\end{lemma}

This count is the rank-avoidance form of the determinantal bound, closely related to the classical Flanders--Atkinson theory of spaces of matrices with restricted rank~\cite{Flanders1962,AtkinsonLloyd1980}. A linear subspace of $M_{n\times r}$ with no nonzero rank-one matrix has dimension at most $(n-1)(r-1)$. The Supplemental Material~\cite{SM} gives an explicit rank-one-free construction showing that this threshold is sharp as a universal pure-contact guarantee on each nonempty full-effective exact-rank layer. In the present setting, avoiding rank-one matrices would mean that every nonzero projective outcome on the untrusted side prepares a Bob state of rank at least two. Once the untrusted effective dimension exceeds this avoidance bound, some outcome must prepare a pure Bob state.

The asymmetry of Eq.~\eqref{eq:rankbound} is explicit. For a fixed steering direction $A\to B$, the untrusted dimension $m$ helps, whereas the trusted dimension $n$ makes rank-one conditional states harder to force. Thus, in $d\otimes2$, steering from the $d$-dimensional side to the qubit side has the large window $r\le d$, whereas the reverse direction has no nontrivial mixed-rank window when $d>2$. For equal dimensions $m=n=d$, the forced window begins with $r=2$ and, for $d>2$, does not include $r=3$. This asymmetry reflects the directional nature of steering.

\textit{Steering consequence.--} Assume Eq.~\eqref{eq:rankbound}, and choose unit representatives $\alpha_0$ and $\phi$ such that $A:=\bra{\alpha_0}\rho\ket{\alpha_0}=aP_{\phi}$ with $a>0$. For a unit tangent direction $\alpha_1\perp\alpha_0$, write the associated two-dimensional Alice block as $\rho_{01}=\left(\begin{smallmatrix}A&B\\ B^\dagger&D\end{smallmatrix}\right)$, where
$B=(\bra{\alpha_0}\otimes I)\rho(\ket{\alpha_1}\otimes I)$ and $D=\bra{\alpha_1}\rho\ket{\alpha_1}$. The projective vector $\ket{\xi_t}=(\ket{\alpha_0}+t\ket{\alpha_1})/\sqrt{1+t^2}$ prepares $\sigma_t=[A+t(B+B^\dagger)+t^2D]/(1+t^2)$. Let $\beta\in\Ker A=\phi^\perp$ be unit. Positivity gives $B^\dagger\beta=0$, because $\beta$ is a zero direction of the positive corner $A$. Hence, in the Bob slice $\Span\{\phi,\beta\}$, the $\phi$--$\beta$ entry of $\sigma_t$ is $t\langle\phi|B|\beta\rangle+O(t^2)$, whereas the kernel population is $O(t^2)$. Thus a nonzero matrix element $\langle\phi|B|\beta\rangle$ produces first-order tangential motion while the inward filling remains second order.

We call the contact \emph{active} if, along some tangent direction, the first-order block connects the support of $A$ to its kernel, equivalently if
\begin{equation}
    P_{\supp A}BP_{\Ker A}\ne0.
\end{equation}
Otherwise it is \emph{degenerate}. This is exactly the support--kernel condition used in the boundary criterion.

The following boundary-contact certificate is the support--kernel form of the boundary theorem established in Ref.~\cite{ZhangChenBoundary2026}; a detailed proof in the notation used here is provided in the Supplemental Material~\cite{SM}. An active contact produces an $O(t)$ support--kernel displacement, while positivity allows only $O(t^2)$ population in the kernel direction; the shrinking-cap argument turns this mismatch into an obstruction to any LHS model. Here $\Gamma_B$ denotes partial transpose on Bob.
\begin{lemma}[Boundary-contact certificate]
\label{lem:boundarycertificate}
An active pure contact for $A\to B$ implies
\begin{equation}
    \rho^{\Gamma_B}\not\ge0,
    \qquad
    \rho\text{ is projectively steerable from }A\text{ to }B.
    \label{eq:mechanism}
\end{equation}
\end{lemma}
For the NPT conclusion, the same matrix element becomes the off-diagonal entry of a principal block with a zero diagonal entry after partial transpose on Bob, and hence gives a negative determinant.

Together, Lemmas~\ref{lem:rankforced} and \ref{lem:boundarycertificate} establish the active branch. If the forced contact is degenerate, a Schur complement peels off one product layer and transfers the same question to an entangled residual. The proof therefore follows the dichotomy
\begin{equation}
\begin{array}{rcl}
\text{active contact} &\Rightarrow& \text{NPT and projective steering},\\
\text{degenerate contact} &\Rightarrow& \text{Schur peel and recurse}.
\end{array}
\label{eq:dichotomy}
\end{equation}
In rank two the first degenerate peel leaves a pure residual, recovering the original rank-2 branch; in higher rank each degenerate peel lowers the rank and keeps the recursion finite.

In the rank window, a pure contact is forced dimension-theoretically, while an active contact yields projective steering through a local boundary-geometric argument. The latter uses only the two-dimensional trusted slice generated by the support vector and one kernel vector. This global--local separation is why the argument extends beyond two qubits without relying on a Bloch-ellipsoid picture.

\textit{Schur-complement recursion.--} To analyze the degenerate branch, let $\rho$ have effective local supports $\mathbb C^m\otimes\mathbb C^n$, and let
\[
    A:=\bra{\alpha_0}\rho\ket{\alpha_0}=aP_{\phi},\qquad a>0,
\]
be a degenerate pure contact for $A\to B$, with $\alpha_0$ and $\phi$ unit. Write $\mathbb C^m=\mathbb C\alpha_0\oplus\mathcal K_A$, where $\mathcal K_A=\alpha_0^\perp$. For $\eta\in\mathcal K_A$, set
\[
    B_\eta=(\bra{\alpha_0}\otimes I)\rho(\ket{\eta}\otimes I).
\]
Degeneracy means
\[
    P_{\supp A}B_\eta P_{\Ker A}=0
    \quad\text{for every }\eta\in\mathcal K_A.
\]

\begin{lemma}[Degenerate-contact peel]
\label{lem:degenerate_peel_main}
There exist $\gamma\in\mathcal K_A$ and a positive operator $\rho^{(1)}$, supported on $\mathcal K_A\otimes\mathbb C^n$, such that
\begin{equation}
    \rho=a\ketbra{\chi}\otimes P_{\phi}+\rho^{(1)},
    \qquad
    \chi=\alpha_0+\gamma.
    \label{eq:main_peel}
\end{equation}
Writing $\rho_A^{(1)}=\Tr_B\rho^{(1)}$, the residual satisfies
\begin{equation}
    \rank\rho^{(1)}=\rank\rho-1,
    \qquad
    \dim\supp\rho_A^{(1)}=m-1.
    \label{eq:main_update}
\end{equation}
If $\rho$ is entangled, then $\rho^{(1)}$ is entangled.
\end{lemma}

\emph{Sketch of proof.} Positivity of the two-by-two block determined by $\alpha_0$ and a tangent vector $\eta$ first gives $B_\eta^\dagger\beta=0$ for every $\beta\perp\phi$, hence $\Ran B_\eta\subseteq\mathbb C\phi$.  The degeneracy assumption also gives $B_\eta P_{\phi^\perp}=0$, so $B_\eta$ is a scalar multiple of $P_{\phi}$.  Thus the whole off-diagonal coupling from the contact corner lies in the product line $\mathcal K_A\otimes\mathbb C\phi$.  Completing the square in the one-dimensional corner $aP_{\alpha_0\otimes\phi}$ gives Eq.~\eqref{eq:main_peel}.  The first term is a positive rank-one product operator, and the Schur complement is positive and supported on $\mathcal K_A\otimes\mathbb C^n$.  The two ranges intersect trivially. Since the peeled product vector has a nonzero $\alpha_0$ component, the Alice support decomposes as
\[
    \supp\rho_A
    =\mathbb C\chi\dotplus\supp\rho_A^{(1)},
\]
so effective support gives Eq.~\eqref{eq:main_update}.  Finally, if $\rho^{(1)}$ were separable, then Eq.~\eqref{eq:main_peel} would be a separable decomposition of $\rho$.  Details, including the lifting of active contacts back through the peel, are given in the Supplemental Material~\cite{SM}. \hfill$\square$

The rank-forcing condition is stable under this peel.  If
\begin{equation}
    m-1\ge(n-1)(r-1),
    \qquad r=\rank\rho,
    \label{eq:main_rank_window}
\end{equation}
then after a degenerate peel the residual has $r_1=r-1$, $m_1=m-1$, and trusted effective dimension $n_1\le n$. If the residual is not pure, its entanglement gives $m_1,n_1\ge2$, so the rank-forcing count may be reapplied. Hence
\begin{align}
    m_1-1-(n_1-1)(r_1-1)
    &= [m-1-(n-1)(r-1)]\notag\\
    &\quad +(n-2)\notag\\
    &\quad +(n-n_1)(r-2)\ge0.
\end{align}
Thus a new pure contact is again forced unless the residual has already become pure.

\begin{theorem}[Directional rank-forced steering]
\label{thm:rankwindow}
Let $\rho$ be an entangled state on effective local supports $\mathbb C^m\otimes\mathbb C^n$, with $m,n\ge2$ and rank $r$.  If
\begin{equation}
    r\le 1+\left\lfloor\frac{m-1}{n-1}\right\rfloor,
    \label{eq:rankwindow_main}
\end{equation}
then $\rho$ is projectively steerable from $A$ to $B$.  Exchanging the parties gives the symmetric statement: if
\begin{equation}
    r\le 1+\left\lfloor\frac{n-1}{m-1}\right\rfloor,
\end{equation}
then $\rho$ is projectively steerable from $B$ to $A$.
\end{theorem}

\emph{Proof.}  We prove the first direction.  If $r=1$, entanglement means that $\rho$ is a pure entangled state, and a pair of nonzero Schmidt coefficients gives a pure contact with an active tangent direction.  Lemma~\ref{lem:boundarycertificate} then gives projective steering.  If $r>1$, Lemma~\ref{lem:rankforced} gives a pure contact.  If this contact is active, Lemma~\ref{lem:boundarycertificate} again applies.  If it is degenerate, Lemma~\ref{lem:degenerate_peel_main} peels off one product layer and leaves an entangled residual satisfying the same rank-forcing inequality.  Iterating, the rank strictly decreases at every degenerate step.  The process therefore terminates after finitely many steps at either an active contact or a pure entangled residual; in the latter case, a Schmidt-pair tangent supplies an active contact.  The active contact lifts through all previous peels: the linear maps $Jx=x-\langle\gamma|x\rangle\alpha_0$ preserve unnormalized blocks, and normalization plus tangent orthogonalization preserve support--kernel nonvanishing.  Hence the original state is projectively steerable. \hfill$\square$

Setting $\rank\rho=2$, if $m\ge n$, then Eq.~\eqref{eq:rankwindow_main} holds and Theorem~\ref{thm:rankwindow} gives steering from $A$ to $B$.  If $n\ge m$, the exchanged statement gives steering from $B$ to $A$.  For equal effective dimensions both conditions hold, proving Theorem~\ref{thm:ranktwo}.  The general criterion therefore places the rank-two theorem at the base of the broader directional hierarchy.

\textit{Certificate and rank window.--} The proof leaves a compact certificate beyond the rank-2 theorem.  A forced pure contact gives unit vectors $\alpha_0$ and $\phi$ with $\bra{\alpha_0}\rho\ket{\alpha_0}=aP_{\phi}$.  A tangent direction $\alpha_1$ is active precisely when some scalar $c_\beta=\langle\phi|B_{\alpha_1}|\beta\rangle$, with $\beta$ a unit vector orthogonal to $\phi$, is nonzero.  The two-dimensional Bob compression onto $\Span\{\phi,\beta\}$ already contains the boundary-steering obstruction, and the same scalar gives a negative $2\times2$ partial-transpose minor.  Thus the certificate consists of a contact vector, a tangent vector, and a kernel direction.

The rank-forced window also has no PPT entanglement.  The low-rank PPT separability theorem says that a PPT state of rank at most $\max\{m,n\}$ is separable~\cite{Peres1996,Horodecki1996,HorodeckiRMP2009,HorodeckiLowRank2000}.  For $m, n \geq 2$, the directional window satisfies
\begin{equation}
    r\le 1+\left\lfloor\frac{m-1}{n-1}\right\rfloor\le m\le\max\{m,n\},
\end{equation}
so every entangled state in this window is NPT.  The nonzero support--kernel tangent block above supplies the steering conclusion.  For $d\otimes2$ the direction $d\to2$ gives $r\le d$; for $d\otimes d$, Theorem~\ref{thm:ranktwo} gives two-way projective steering at rank two.

\textit{Discussion and outlook.--} Theorem~\ref{thm:ranktwo} identifies rank two as the first genuine mixed-state test of the entanglement--steering hierarchy and settles it.  Once the steering orientation is chosen according to the effective local dimensions, there is no exceptional rank-2 entangled family with a projective LHS model.  The result covers the full rank-two class: the first possible essential separation between entanglement and projective steering is pushed beyond rank two, and measurement optimization disappears from the first mixed-state rank.

Beyond the rank-2 theorem, the proof also leaves a usable certificate.  First pass to effective local supports and choose the steering direction.  Second, check the rank window and solve the rank-one contraction problem for $L=\{V(x):x\in(\mathbb C^m)^*\}$.  This identifies a product-zero slice $\alpha\otimes\phi^\perp\subseteq\Ker\rho$ and hence a pure trusted conditional state.  Third, for a unit tangent direction $\alpha_1\perp\alpha_0$, inspect the neighboring block $B_{\alpha_1}:=(\bra{\alpha_0}\otimes I)\rho(\ket{\alpha_1}\otimes I)$ for a nonzero scalar $\langle\phi|B_{\alpha_1}|\beta\rangle$.  Such a scalar produces a negative partial-transpose minor and certifies projective steering.  Low-rank linear algebra is thereby converted directly into steering.

The argument applies across arbitrary finite dimensions and general low-rank families.  It uses only the support and kernel of the density operator, the parts of a low-rank state that are often fixed by preparation, reconstruction, variational structure, dissipative dark-state manifolds, or the chosen bipartition of a larger system.  In such settings structural data are typically more robust than an optimized steering inequality and cheaper to use than a full SDP search.  The present theorem shows that, at rank two, this robust information already fixes the projective-steering status of every entangled state.

The conceptual message is that steering can be forced by how a state meets the boundary of the trusted positive cone.  At an active contact, positivity supplies only a quadratic inward filling while quantum mechanics supplies a linear tangential coherence; an LHS model would have to concentrate hidden-state weight in caps shrinking to a point.  At rank two, the global rank constraint and the local boundary obstruction lock together completely, turning entanglement itself into projective steering. Extensions to multipartite settings remain an open direction, where the boundary geometry is considerably richer and a different structural analysis --- beyond the two-body Schur recursion --- would be required.


\emph{Acknowledgments.---}This work is supported by the Quantum Science and Technology-National Science and Technology Major Project (Grant No. 2024ZD0301000), and the National Natural Science Foundation of China (Grant No. 12275136).

\end{document}


\title{Supplemental Material for ``Every Rank-2 Bipartite Entangled State is Projectively Steerable''}

\author{Yu-Xuan Zhang}
\affiliation{School of Physics, Nankai University, Tianjin 300071, People's Republic of China}

\author{Jing-Ling Chen}
\email{chenjl@nankai.edu.cn}
\affiliation{Theoretical Physics Division, Chern Institute of Mathematics, Nankai University, Tianjin 300071, People's Republic of China}

\date{\today}

\maketitle

This Supplemental Material provides the technical details behind the Letter.  Sec.~S1 fixes the projective-assemblage and effective-support conventions.  Sec.~S2 proves the rank-forced pure-contact lemma and the sharpness of its universal pure-contact window.  Secs.~S3 and S4 establish the boundary-contact steering certificate, including the associated NPT minor.  Sec.~S5 develops the Schur-complement peel for degenerate contacts and identifies the rank-two branch as its first case.  Sec.~S6 gives the finite rank-forcing recursion used in the directional theorem.  Sec.~S7 explains the consequences of the rank window, its relation to PPT separability, and the use of the support--kernel certificate beyond the forced regime.  The companion work~\cite{ZhangChenBoundary2026} develops the boundary geometry in a broader setting.

\section{S1. Conventions and effective supports}
\label{sec:conventions}

Let $\rho$ be a positive operator on $\mathcal H_A\otimes\mathcal H_B$.  For the steering direction $A\to B$, Alice is untrusted and Bob is trusted.  A rank-one projective outcome with unit vector $\alpha\in\mathcal H_A$ prepares the unnormalized Bob state
\begin{equation}
    \sigma_\alpha=\Tr_A[(P_{\alpha}\otimes I_B)\rho]=\bra{\alpha}\rho\ket{\alpha}.
\end{equation}
When $\Tr\sigma_\alpha>0$, the normalized conditional state is $\sigma_\alpha/\Tr\sigma_\alpha$; the assemblage itself is the collection of unnormalized conditional states obtained from all complete rank-one projective measurements on the untrusted side.  A state is projectively steerable from $A$ to $B$ if this projective assemblage has no local-hidden-state (LHS) representation.  Throughout, $\ket{x}\bra{x}$ denotes the unnormalized rank-one operator associated with $x$.  If $S$ is a subspace, $P_S$ denotes the orthogonal projection onto $S$; for a nonzero vector $x$, we write $P_x:=P_{\mathbb Cx}$.  Thus $\ket{\alpha}\otimes\phi^\perp\subseteq\Ker\rho$ means $\ket{\alpha}\otimes\ket{\beta}\in\Ker\rho$ for every $\beta\perp\phi$.  Vectors that label actual projective outcomes are taken to be unit representatives.

The effective Alice support is $\supp\rho_A$, where $\rho_A=\Tr_B\rho$, and similarly for Bob.  All steering statements in the Letter are made after restricting to the tensor product of the two effective supports.  This restriction preserves all nonzero conditional states.  Indeed, if $\alpha\perp\supp\rho_A$, then $\bra{\alpha}\rho_A\ket{\alpha}=0$, so positivity gives $(\bra{\alpha}\otimes I)\rho^{1/2}=0$ and hence $\sigma_\alpha=0$.  Such a vector is a zero-probability outcome for every measurement and carries no steering information.

We use the elementary fact that if $M\ge0$ and $\langle v|M|v\rangle=0$, then $M\ket{v}=0$.  It follows by writing $M=X^\dagger X$ and observing that $\|Xv\|^2=0$.

\section{S2. Rank-forced pure contacts}
\label{sec:rankforcing}

Let $r=\rank\rho$ and choose a square-root decomposition
\begin{equation}
    \rho=\sum_{j=1}^{r}\ketbra{\psi_j},
\end{equation}
where the vectors $\psi_j\in\mathbb C^m\otimes\mathbb C^n$ are linearly independent.  For a covector $x\in(\mathbb C^m)^*$, define the $n\times r$ matrix
\begin{equation}
    V(x)=\big[(x\otimes I)\psi_1,\ldots,(x\otimes I)\psi_r\big].
\end{equation}
For a unit representative of a projective outcome $\alpha$, we evaluate this linear map at $x=\bra{\alpha}$. Then
\begin{equation}
    \sigma_\alpha=V(\bra{\alpha})V(\bra{\alpha})^\dagger,
    \qquad
    \rank\sigma_\alpha=\rank V(\bra{\alpha}).
\end{equation}
The map $x\mapsto V(x)$ is linear.  The identification of a ket $\alpha$ with $\bra{\alpha}$ is antilinear, but the projective dimension count is carried out in the dual contraction space.  On the effective Alice support it is injective: if $V(\bra{\alpha})=0$, then $\sigma_\alpha=0$, hence $\bra{\alpha}\rho_A\ket{\alpha}=0$, which contradicts $\alpha\in\supp\rho_A$ unless $\alpha=0$.  Therefore its image $L\subset M_{n\times r}(\mathbb C)$ has dimension $m$.

A nonzero projective outcome prepares a pure Bob conditional state precisely when $V(\bra{\alpha})$ has rank one.  The lemma below records the dimension count used in the Letter.

\begin{lemma}[Rank-forced pure contact]
\label{lem:rank_forced_SM}
For the steering direction $A\to B$, let $m$ and $n$ be the effective local dimensions and let $r=\rank\rho$.  If
\begin{equation}
    r\le 1+\left\lfloor\frac{m-1}{n-1}\right\rfloor,
    \label{eq:SM_rank_bound}
\end{equation}
then there exists a projective outcome with unit representative $\alpha$ such that $\sigma_\alpha$ is nonzero and rank one.  Equivalently, for some Bob line $\mathbb C\phi$,
\begin{equation}
    \ket{\alpha}\otimes\phi^\perp\subseteq\Ker\rho.
    \label{eq:SM_tensor_zero}
\end{equation}
\end{lemma}

\emph{Proof.}  Over $\mathbb C$, in projective space $\mathbb P(M_{n\times r})$, the nonzero rank-one matrices form the Segre variety
\begin{equation}
    \Sigma_1=\mathbb P^{n-1}\times\mathbb P^{r-1},
    \qquad
    \dim\Sigma_1=n+r-2.
\end{equation}
The projectivization $\mathbb P L$ has dimension $m-1$, and the ambient projective dimension is $nr-1$.  The projective dimension theorem~\cite{Harris1992} gives
\begin{equation}
    \dim(\mathbb P L\cap\Sigma_1)
    \ge (m-1)+(n+r-2)-(nr-1).
\end{equation}
Thus $\mathbb P L$ must meet $\Sigma_1$ whenever
\begin{equation}
    (m-1)+(n+r-2)\ge nr-1,
\end{equation}
which is equivalent to $m>(n-1)(r-1)$, or to Eq.~\eqref{eq:SM_rank_bound}.  Hence some nonzero $V(\bra{\alpha})$ has rank one. Choosing a unit representative of the corresponding Alice line, $\sigma_\alpha=V(\bra{\alpha})V(\bra{\alpha})^\dagger$ is nonzero and rank one.

We next relate this to the tensor-zero form.  If $\sigma_\alpha=aP_{\phi}$ with $a>0$ and $\phi$ is chosen unit, then every $\beta\perp\phi$ satisfies
\begin{equation}
    0=\langle\beta|\sigma_\alpha|\beta\rangle
    =\langle\alpha,\beta|\rho|\alpha,\beta\rangle.
\end{equation}
Since $\rho\ge0$, this implies $\ket{\alpha}\otimes\ket{\beta}\in\Ker\rho$.  Hence Eq.~\eqref{eq:SM_tensor_zero} holds.  Conversely, Eq.~\eqref{eq:SM_tensor_zero} says that the range of $\sigma_\alpha$ is contained in $\mathbb C\phi$; because $\alpha$ is in the effective support, $\sigma_\alpha\ne0$, so it has rank one. \hfill$\square$

\paragraph*{Directional character.}  The bound is intrinsically directional.  For $A\to B$, the untrusted dimension $m$ enlarges the contraction space, whereas the trusted dimension $n$ enlarges the determinantal variety that must be intersected.  Exchanging the parties gives the bound with $m$ and $n$ interchanged.  In particular, for rank two the condition becomes $m\ge n$ for $A\to B$, while the reverse direction requires $n\ge m$.

\begin{lemma}[Sharpness of the pure-contact window]
\label{lem:SM_sharp_pure_contact_window}
Fix the steering direction $A\to B$, and let $m=\dim A$ and $n=\dim B$ be the effective dimensions.  Assume that full-effective states of exact rank $r$ exist.  If
\begin{equation}
    m\le (n-1)(r-1),
    \label{eq:SM_sharp_contact_window}
\end{equation}
then there exists a rank-$r$ state on the effective $m\times n$ system with no $A\to B$ pure contact.
\end{lemma}

\emph{Proof.}  Let
\begin{equation}
    M=\operatorname{Hom}(\mathbb C^r,\mathbb C^n)\cong M_{n\times r}(\mathbb C).
\end{equation}
The projectivization of the nonzero rank-one matrices is the Segre variety
\begin{equation}
    \Sigma_1\simeq \mathbb P^{n-1}\times\mathbb P^{r-1}
    \subset \mathbb P(M),
\end{equation}
whose codimension in $\mathbb P(M)$ is
\begin{equation}
    (nr-1)-(n+r-2)=(n-1)(r-1).
\end{equation}

We first construct a rank-one-free linear subspace of this critical dimension.  Index rows by $0,\ldots,n-1$ and columns by $0,\ldots,r-1$, and set
\begin{equation}
    S=\left\{W=(W_{ij})\in M:
    \sum_{i+j=k}W_{ij}=0\quad\text{for every }k=0,\ldots,n+r-2\right\}.
\end{equation}
The $n+r-1$ anti-diagonal equations are independent, so
\begin{equation}
    \dim S=nr-(n+r-1)=(n-1)(r-1).
\end{equation}
If $W=uv^T\in S$ has rank at most one, then the anti-diagonal sums are the coefficients of
\begin{equation}
    \left(\sum_{i=0}^{n-1}u_i t^i\right)
    \left(\sum_{j=0}^{r-1}v_j t^j\right).
\end{equation}
All these coefficients vanish.  Since $\mathbb C[t]$ has no zero divisors, $u=0$ or $v=0$, and hence $W=0$.  Thus $S$ contains no nonzero rank-one matrix.  Since $m\le\dim S$, any $m$-plane inside $S$ is rank-one-free.

Let $\operatorname{Gr}(m,M)$ be the Grassmannian of $m$-planes in $M$.  We now impose the two full-effective support conditions in addition to rank-one-freeness.

First, the condition
\begin{equation}
    \mathbb P L\cap\Sigma_1=\varnothing
\end{equation}
is Zariski open in $\operatorname{Gr}(m,M)$.  Indeed, the incidence set
\begin{equation}
    \mathcal I_1=
    \{(L,[W])\in \operatorname{Gr}(m,M)\times\Sigma_1:
    W\in L\}
\end{equation}
is closed, and its projection to $\operatorname{Gr}(m,M)$ is closed because $\Sigma_1$ is projective.  Therefore the complement is open.  This open set is nonempty by the $m$-planes contained in $S$ constructed above.

Next consider the condition
\begin{equation}
    \bigcap_{W\in L}\ker W=\{0\}.
\end{equation}
Its failure is equivalent to the existence of a nonzero vector $c\in\mathbb C^r$ such that $Wc=0$ for all $W\in L$.  Hence the bad locus is the projection of the closed incidence set
\begin{equation}
    \mathcal I_{\ker}=
    \{(L,[c])\in \operatorname{Gr}(m,M)\times\mathbb P^{r-1}:
    Wc=0\text{ for all }W\in L\},
\end{equation}
and is closed.  Thus the condition is Zariski open.  It is nonempty: since an exact rank-$r$ state on $A\otimes B$ exists, $r\le mn$.  Therefore one can choose $m$ matrices $W_1,\ldots,W_m\in M$ such that the vertically stacked map
\begin{equation}
    c\longmapsto (W_1c,\ldots,W_mc)\in (\mathbb C^n)^m
\end{equation}
is injective.  The set of such $m$-tuples is a nonempty Zariski open subset of $M^m$, and the set of linearly independent $m$-tuples is also a nonempty Zariski open subset of $M^m$.  Since $M^m$ is irreducible, their intersection is nonempty.  The span of any tuple in this intersection gives an $m$-plane satisfying the common-kernel condition.

Finally consider the condition
\begin{equation}
    \Span\{\Ran W:W\in L\}=\mathbb C^n.
\end{equation}
Its failure is equivalent to the existence of a nonzero covector $\eta\in(\mathbb C^n)^*$ such that $\eta W=0$ for all $W\in L$.  Hence the bad locus is the projection of the closed incidence set
\begin{equation}
    \mathcal I_{\Ran}=
    \{(L,[\eta])\in \operatorname{Gr}(m,M)\times\mathbb P((\mathbb C^n)^*):
    \eta W=0\text{ for all }W\in L\},
\end{equation}
and is closed.  Thus the condition is Zariski open.  It is nonempty because the assumed nonemptiness of the full-effective exact-rank layer implies that Bob can have full effective support at rank $r$.  Since the Bob support of a rank-$r$ state is spanned by the Bob supports of $r$ vectors in $A\otimes B$, each of dimension at most $m$, one must have $n\le mr$.  Therefore one can choose $m$ matrices $W_1,\ldots,W_m\in M$ such that the horizontal block matrix
\begin{equation}
    [\,W_1\ W_2\ \cdots\ W_m\,]
\end{equation}
has rank $n$.  As above, intersecting this nonempty Zariski open condition in $M^m$ with linear independence gives an $m$-plane satisfying the full-range condition.

We have therefore obtained three nonempty Zariski open subsets of the irreducible variety $\operatorname{Gr}(m,M)$.  Their intersection is nonempty.  Choose $L\subset M$ in this intersection, so that
\begin{equation}
    \mathbb P L\cap\Sigma_1=\varnothing,
    \qquad
    \bigcap_{W\in L}\ker W=\{0\},
    \qquad
    \Span\{\Ran W:W\in L\}=\mathbb C^n.
    \label{eq:SM_sharp_L_properties}
\end{equation}

Choose a linear isomorphism $F:A^*\to L$.  For $j=1,\ldots,r$, define $\psi_j\in A\otimes B$ by requiring
\begin{equation}
    (x\otimes I)\psi_j=F(x)e_j,
    \qquad x\in A^*,
\end{equation}
where $e_1,\ldots,e_r$ is the standard basis of $\mathbb C^r$.  This is well-defined because the right-hand side is a linear map from $A^*$ to $B$.  Set
\begin{equation}
    \rho=\sum_{j=1}^{r}\ketbra{\psi_j}.
\end{equation}

We first check the rank.  If $\sum_j c_j\psi_j=0$, then for every $x\in A^*$,
\begin{equation}
    F(x)c=0.
\end{equation}
Since $F(A^*)=L$, Eq.~\eqref{eq:SM_sharp_L_properties} gives $c=0$.  Hence $\psi_1,\ldots,\psi_r$ are linearly independent, and $\rank\rho=r$.

The Bob effective support is all of $B$, because
\begin{equation}
\begin{split}
    \Span\{(x\otimes I)\psi_j:x\in A^*,\ 1\le j\le r\}
    &=\Span\{F(x)e_j:x\in A^*,\ 1\le j\le r\}  \\
    &=\Span\{\Ran W:W\in L\}
    =B.
\end{split}
\end{equation}
The Alice effective support is all of $A$.  Indeed, for every nonzero $x\in A^*$ one has $F(x)\ne0$, and the corresponding conditional state satisfies
\begin{equation}
    \sigma_x=F(x)F(x)^\dagger,
    \qquad
    \Tr\sigma_x=\|F(x)\|_{\mathrm{HS}}^2>0.
\end{equation}
Thus no nonzero Alice direction is removed by passing to effective support.

Finally, let $\alpha$ be any nonzero Alice outcome and put $x=\bra{\alpha}$.  Then
\begin{equation}
    \sigma_\alpha=F(x)F(x)^\dagger.
\end{equation}
Since $0\ne F(x)\in L$ and $L$ contains no nonzero rank-one matrix, $\rank F(x)\ne1$.  Hence
\begin{equation}
    \rank\sigma_\alpha=\rank F(x)\ne1.
\end{equation}
Therefore $\rho$ has no $A\to B$ pure contact.  Normalizing $\rho$ by its trace does not change its rank, effective supports, or conditional-state ranks, so the constructed positive operator gives the claimed rank-$r$ state.  This proves that the inequality in Lemma~\ref{lem:rank_forced_SM} is sharp as a universal pure-contact criterion whenever full-effective states of exact rank $r$ exist. \hfill$\square$

\section{S3. The two-dimensional boundary-scaling obstruction}
\label{sec:scaling}

The steering obstruction used in the Letter is local near a pure trusted boundary point.  We use the two-dimensional matrix form needed after the trusted compression in Sec.~S4, which makes the support--kernel block explicit.

Let $\sigma_t$, $0<t<t_0$, be unnormalized conditional states on a two-dimensional trusted space with ordered basis $\{e_0,e_1\}$.  Assume that $\sigma_t$ approaches the pure boundary ray $\mathbb C e_0$ and write
\begin{equation}
    \sigma_t=
    \begin{pmatrix}
        a_t & b_t\\
        \bar b_t & d_t
    \end{pmatrix},
    \qquad a_t,d_t\ge0.
\end{equation}
Suppose that, for all sufficiently small $t>0$,
\begin{equation}
    |b_t|\ge \ell t,
    \qquad
    d_t\le c t^2,
    \label{eq:SM_scaling_assumption}
\end{equation}
with constants $\ell>0$ and $c<\infty$.  The entry $b_t$ is the first-order support--kernel coherence, while $d_t$ is the population in the kernel direction.  The obstruction is that an LHS model cannot create a linear off-diagonal term while keeping the kernel population quadratic.

\begin{lemma}[Boundary-scaling obstruction]
\label{lem:SM_cap_scaling}
No LHS representation can reproduce a family satisfying Eq.~\eqref{eq:SM_scaling_assumption}.  Consequently any projective assemblage containing such a one-parameter subfamily is not LHS.
\end{lemma}

\emph{Proof.}  Assume, toward a contradiction, that there is an LHS representation of this family.  For the present one-parameter family of projective outcomes, an LHS model has the form
\begin{equation}
    \sigma_t=\int g_t(\lambda)\,\omega_\lambda\,d\mu(\lambda),
    \qquad 0\le g_t(\lambda)\le1,
    \label{eq:SM_lhs_lambda_space}
\end{equation}
with a finite hidden-state measure $\mu$ and normalized trusted states $\omega_\lambda$ on $\Span\{e_0,e_1\}$.  Push $\mu$ forward under $\lambda\mapsto \omega_\lambda$ to a finite measure $\nu$ on the compact set of normalized density matrices on $\Span\{e_0,e_1\}$.  For each $t$, the response-weighted measure $g_t(\lambda)\,d\mu(\lambda)$ pushes forward to a measure dominated by $\nu$; by the Radon--Nikodym theorem there exists a measurable density $f_t(\tau)$, with $0\le f_t(\tau)\le1$ $\nu$-almost everywhere, such that
\begin{equation}
    \sigma_t=\int f_t(\tau)\,\tau\,d\nu(\tau),
    \qquad 0\le f_t(\tau)\le1,
    \label{eq:SM_lhs_density_space}
\end{equation}
where the dependence on $t$ is carried by the response density $f_t$.  In the same ordered basis write
\begin{equation}
    \tau=
    \begin{pmatrix}
        a(\tau)&b(\tau)\\
        \overline{b(\tau)}&d(\tau)
    \end{pmatrix},
    \qquad \tau\ge0,
    \qquad \Tr\tau=1.
\end{equation}
Positivity gives
\begin{equation}
    |b(\tau)|^2\le a(\tau)d(\tau)\le d(\tau).
    \label{eq:SM_entry_bound}
\end{equation}
Moreover,
\begin{equation}
    b_t=\int f_t(\tau)b(\tau)\,d\nu(\tau),
    \qquad
    d_t=\int f_t(\tau)d(\tau)\,d\nu(\tau)
    \le c t^2.
    \label{eq:SM_entry_integrals}
\end{equation}

Choose $K>2c/\ell$ and define the shrinking cap around the boundary point $N=P_{e_0}$ by
\begin{equation}
    \Gamma_t=\{\tau: d(\tau)\le K^2t^2\}.
\end{equation}
On the complement $\Gamma_t^c$, Eq.~\eqref{eq:SM_entry_bound} and $d(\tau)>K^2t^2$ imply
\begin{equation}
    |b(\tau)|\le \frac{d(\tau)}{Kt}.
\end{equation}
Therefore
\begin{equation}
\begin{split}
    \left|\int_{\Gamma_t^c} f_t(\tau)b(\tau)\,d\nu(\tau)\right|
    &\le \frac{1}{Kt}\int_{\Gamma_t^c} f_t(\tau)d(\tau)\,d\nu(\tau)  \\
    &\le \frac{c}{K}t < \frac{\ell}{2}t .
\end{split}
\end{equation}
Since $|b_t|\ge \ell t$, the cap contribution must satisfy
\begin{equation}
    \left|\int_{\Gamma_t} f_t(\tau)b(\tau)\,d\nu(\tau)\right|
    \ge \frac{\ell}{2}t.
    \label{eq:SM_cap_contribution}
\end{equation}
Inside $\Gamma_t$, Eq.~\eqref{eq:SM_entry_bound} gives $|b(\tau)|\le Kt$.  At the boundary point $N=P_{e_0}$ one has $d(N)=0$ and $b(N)=0$, so an atom at $N$ cannot contribute to the off-diagonal integral.  Hence Eq.~\eqref{eq:SM_cap_contribution} implies
\begin{equation}
    \nu(\Gamma_t\setminus\{N\})\ge \frac{\ell}{2K}
    \label{eq:SM_positive_cap_mass}
\end{equation}
for all sufficiently small $t>0$.

This is impossible: along any sequence $t_j\downarrow0$, the sets $\Gamma_{t_j}\setminus\{N\}$ decrease to the empty set: if $d(\tau)=0$, positivity and $\Tr\tau=1$ force $\tau=N$.  By continuity from above for finite measures, $\nu(\Gamma_{t_j}\setminus\{N\})\to0$, contradicting Eq.~\eqref{eq:SM_positive_cap_mass}. \hfill$\square$

\paragraph*{Interpretation.}  The obstruction concerns the joint scaling of the one-parameter family, rather than the decomposition of any single conditional state.

\section{S4. Support--kernel boundary certificate}
\label{sec:sk}

The operator criterion below is the two-dimensional support--kernel form of the boundary theorem developed in Ref.~\cite{ZhangChenBoundary2026}.  Let $X$ be the untrusted measuring side and $Y$ the trusted side.  Choose orthonormal vectors $\ket{\alpha_0},\ket{\alpha_1}\in\mathcal H_X$ and write the corresponding $2\times2$ block of $\rho$ as
\begin{equation}
    \rho_{01}=\begin{pmatrix}A&B\\B^\dagger&D\end{pmatrix}_X,
    \quad
    A=\bra{\alpha_0}\rho\ket{\alpha_0},
    \quad
    B=(\bra{\alpha_0}\otimes I)\rho(\ket{\alpha_1}\otimes I),
    \quad
    D=\bra{\alpha_1}\rho\ket{\alpha_1},
\end{equation}
where $A,B,D$ act on $\mathcal H_Y$.  The projective outcome
\begin{equation}
    \ket{\xi_t}=\frac{\ket{\alpha_0}+t\ket{\alpha_1}}{\sqrt{1+t^2}}
\end{equation}
prepares the unnormalized trusted state
\begin{equation}
    \sigma_t=\frac{A+t(B+B^\dagger)+t^2D}{1+t^2}.
    \label{eq:SM_trusted_path}
\end{equation}

\begin{proposition}[Support--kernel boundary certificate]
\label{prop:SM_sk_certificate}
Assume that $\rho$ is a state, $A\ne0$, $\Ker A\ne\{0\}$, and
\begin{equation}
    P_{\supp A}BP_{\Ker A}\ne0.
    \label{eq:SM_sk_condition}
\end{equation}
Then $\rho^{\Gamma_Y}\not\ge0$, and the projective assemblage generated by complete rank-one projective measurements on $X$ is not LHS.  Hence $\rho$ is projectively steerable from $X$ to $Y$.
\end{proposition}

\emph{Proof.}  Choose unit vectors $\phi\in\supp A$ and $\beta\in\Ker A$ such that
\begin{equation}
    v=\langle\phi|B|\beta\rangle\ne0.
\end{equation}
Since $\beta\in\Ker A$,
\begin{equation}
    \langle\alpha_0,\beta|\rho|\alpha_0,\beta\rangle
    =\langle\beta|A|\beta\rangle=0.
\end{equation}
Positivity of $\rho$ therefore gives $\rho\ket{\alpha_0,\beta}=0$, and in particular $B^\dagger\beta=0$.

Let $\mathcal T=\Span\{\phi,\beta\}$ and let $P_{\mathcal T}$ be the orthogonal projection onto $\mathcal T$.  Compress Eq.~\eqref{eq:SM_trusted_path} to $\mathcal T$.  In the ordered basis $\{\phi,\beta\}$, the compressed conditional state $\widetilde\sigma_t=P_{\mathcal T}\sigma_tP_{\mathcal T}$ has matrix entries
\begin{equation}
    \widetilde\sigma_t=
    \begin{pmatrix}
        a+O(t) & tv+O(t^2)\\
        t\bar v+O(t^2) & O(t^2)
    \end{pmatrix},
    \label{eq:SM_compressed_scaling}
\end{equation}
where $a=\langle\phi|A|\phi\rangle>0$ and $d=\langle\beta|D|\beta\rangle\ge0$.  The relation $B^\dagger\beta=0$ removes the first-order lower-right term.  Thus, for all sufficiently small $t>0$, the off-diagonal entry is bounded below linearly in $t$, while the $\beta$-diagonal entry is bounded above quadratically in $t$.  Lemma~\ref{lem:SM_cap_scaling} applies to the compressed family.

We next show that this compressed obstruction is inherited by the original assemblage.  Suppose that the original $X\to Y$ projective assemblage had an LHS model, so that
\begin{equation}
    \sigma_t=\int f_t(\lambda)\tau_\lambda\,d\mu(\lambda),
    \qquad 0\le f_t(\lambda)\le1,
    \label{eq:SM_original_lhs}
\end{equation}
with $\mu$ finite and each $\tau_\lambda$ a normalized state on $\mathcal H_Y$.  Compressing gives
\begin{equation}
    \widetilde\sigma_t=\int f_t(\lambda)P_{\mathcal T}\tau_\lambda P_{\mathcal T}\,d\mu(\lambda).
\end{equation}
For those $\lambda$ with $q_\lambda=\Tr(P_{\mathcal T}\tau_\lambda P_{\mathcal T})>0$, define
\begin{equation}
    \widetilde\tau_\lambda=\frac{P_{\mathcal T}\tau_\lambda P_{\mathcal T}}{q_\lambda},
    \qquad
    d\widetilde\mu(\lambda)=q_\lambda\,d\mu(\lambda).
\end{equation}
Since $0\le q_\lambda\le1$, the measure $\widetilde\mu$ is finite.  Points with $q_\lambda=0$ contribute zero to every compressed matrix entry and may be discarded.  Hence
\begin{equation}
    \widetilde\sigma_t=\int f_t(\lambda)\widetilde\tau_\lambda\,d\widetilde\mu(\lambda)
\end{equation}
is an LHS representation on the two-dimensional trusted space $\mathcal T$.  This contradicts Lemma~\ref{lem:SM_cap_scaling}.  Therefore the original projective assemblage is not LHS.

For the NPT statement, take the partial transpose on $Y$ in a basis containing $\phi$ and $\beta$, and restrict $\rho^{\Gamma_Y}$ to the subspace
\begin{equation}
    \Span\{\ket{\alpha_0,\beta},\ket{\alpha_1,\phi}\}.
\end{equation}
The first diagonal entry is $\langle\beta|A|\beta\rangle=0$, and the off-diagonal entry is
\begin{equation}
    \langle\alpha_0,\beta|\rho^{\Gamma_Y}|\alpha_1,\phi\rangle
    =\langle\alpha_0,\phi|\rho|\alpha_1,\beta\rangle
    =v.
\end{equation}
Thus the principal block has the form
\begin{equation}
    \begin{pmatrix}0&v\\\bar v&*\end{pmatrix}
\end{equation}
with determinant $-|v|^2<0$.  Hence $\rho^{\Gamma_Y}\not\ge0$. \hfill$\square$

\begin{corollary}[Pure-contact form]
\label{cor:SM_pure_contact}
If $A=aP_{\phi}$ with $a>0$ and $P_{\phi}BP_{\phi^\perp}\ne0$, then the state is NPT and projectively steerable from $X$ to $Y$.
\end{corollary}

\section{S5. Degenerate pure contacts and Schur complements}
\label{sec:SM_peeling}

A degenerate pure contact admits a Schur-complement peel.  In the rank-two case, the residual has rank one.

Fix the direction $A\to B$ and suppose that a projective outcome $\alpha_0$ prepares a pure Bob state.  Choose $\alpha_0$ and $\phi$ unit, so that
\begin{equation}
    A:=\bra{\alpha_0}\rho\ket{\alpha_0}=aP_{\phi},
    \qquad a>0.
\end{equation}
Write $\mathcal H_A=\mathbb C\alpha_0\oplus\mathcal K_A$ with $\mathcal K_A=\alpha_0^\perp$, and, for $\eta\in\mathcal K_A$, set
\begin{equation}
    B_\eta=(\bra{\alpha_0}\otimes I)\rho(\ket{\eta}\otimes I).
\end{equation}
The contact is called degenerate if
\begin{equation}
    P_{\phi} B_\eta P_{\phi^\perp}=0
    \qquad\text{for all }\eta\in\mathcal K_A.
    \label{eq:SM_degenerate_def}
\end{equation}

\begin{lemma}[Shape of degenerate tangent blocks]
\label{lem:SM_degenerate_shape}
If the contact is degenerate, then for every $\eta\in\mathcal K_A$ there is a scalar $b(\eta)$ such that
\begin{equation}
    B_\eta=b(\eta)P_{\phi}.
    \label{eq:SM_scalar_block}
\end{equation}
\end{lemma}

\emph{Proof.}  If $\beta\perp\phi$, then
\begin{equation}
    \langle\alpha_0,\beta|\rho|\alpha_0,\beta\rangle=0.
\end{equation}
Since $\rho\ge0$, we have $\rho\ket{\alpha_0,\beta}=0$.  Hence $B_\eta^\dagger\beta=0$ for every $\beta\perp\phi$, and so $\Ran B_\eta\subseteq\mathbb C\phi$.  Thus $B_\eta=P_{\phi} B_\eta$.  The degeneracy condition gives $P_{\phi} B_\eta P_{\phi^\perp}=0$, hence $B_\eta=B_\eta P_{\phi}$.  Therefore $B_\eta=P_{\phi} B_\eta P_{\phi}$, which proves Eq.~\eqref{eq:SM_scalar_block}. \hfill$\square$

Since $B_\eta$ depends linearly on $\eta$, so does $b(\eta)$.  Hence there exists a vector $\gamma\in\mathcal K_A$ such that
\begin{equation}
    b(\eta)=a\langle\gamma|\eta\rangle.
    \label{eq:SM_gamma_def}
\end{equation}
Let
\begin{equation}
    \chi=\alpha_0+\gamma.
\end{equation}

\begin{lemma}[Schur-complement peel]
\label{lem:SM_peel}
Under the assumptions above,
\begin{equation}
    \rho=a\ketbra{\chi}\otimes P_{\phi}+\rho^{(1)},
    \qquad \rho^{(1)}\ge0,
    \label{eq:SM_peel_decomp}
\end{equation}
where $\rho^{(1)}$ is supported on $\mathcal K_A\otimes\mathcal H_B$.  Moreover
\begin{equation}
    \rank\rho=1+\rank\rho^{(1)}.
    \label{eq:SM_rank_add}
\end{equation}
\end{lemma}

\emph{Proof.}  Decompose the Hilbert space as
\begin{equation}
    \mathbb C(\alpha_0\otimes\phi)
    \oplus(\mathcal K_A\otimes\mathbb C\phi)
    \oplus(\mathcal K_A\otimes\phi^\perp)
    \oplus(\alpha_0\otimes\phi^\perp).
\end{equation}
The last summand is contained in $\Ker\rho$, so $\rho$ has no nonzero matrix elements involving it.  By Lemma~\ref{lem:SM_degenerate_shape}, the off-diagonal row from the one-dimensional contact corner into $\mathcal K_A\otimes\mathcal H_B$ is represented by the vector $a\gamma\otimes\phi$.  Thus, with respect to
\begin{equation}
    \mathbb C(\alpha_0\otimes\phi)
    \oplus(\mathcal K_A\otimes\mathcal H_B),
\end{equation}
we may write
\begin{equation}
    \rho=
    \begin{pmatrix}
        a & a\langle\gamma\otimes\phi| \\
        a|\gamma\otimes\phi\rangle & R
    \end{pmatrix}.
\end{equation}
The Schur complement of the positive scalar corner $a$ is positive:
\begin{equation}
    R-a\ketbra{\gamma\otimes\phi}\ge0.
\end{equation}
Therefore
\begin{equation}
    \rho
    =a\ketbra{\alpha_0\otimes\phi+\gamma\otimes\phi}
      +\left(R-a\ketbra{\gamma\otimes\phi}\right),
\end{equation}
which is Eq.~\eqref{eq:SM_peel_decomp}.  The peeled vector has a nonzero component in the contact corner, while the residual is supported in the complementary subspace $\mathcal K_A\otimes\mathcal H_B$.  Their ranges therefore intersect trivially.  Since for positive semidefinite operators with trivial range intersection one has $\Ran(X+Y)=\Ran X+\Ran Y$, the ranks add, giving Eq.~\eqref{eq:SM_rank_add}. \hfill$\square$

\begin{lemma}[Effective support after a peel]
\label{lem:SM_support_update}
Assume that $\rho$ is effective on Alice, so $\supp\rho_A=\mathbb C\alpha_0\oplus\mathcal K_A$ with $\mathcal K_A=\alpha_0^\perp$.  In the decomposition \eqref{eq:SM_peel_decomp}, if $\rho^{(1)}_A=\Tr_B\rho^{(1)}$, then
\begin{equation}
    \supp\rho_A^{(1)}=\mathcal K_A.
    \label{eq:SM_A_support_update}
\end{equation}
Thus the effective Alice dimension drops by exactly one.
\end{lemma}

\emph{Proof.}  Taking the partial trace over Bob in Eq.~\eqref{eq:SM_peel_decomp} gives
\begin{equation}
    \rho_A=a\ketbra{\chi}+\rho_A^{(1)},
    \qquad \supp\rho_A^{(1)}\subseteq\mathcal K_A.
\end{equation}
Because $\chi=\alpha_0+\gamma$ has a nonzero $\alpha_0$ component, $\chi\notin\mathcal K_A$.  If $\supp\rho_A^{(1)}$ were a proper subspace of $\mathcal K_A$, then the sum of the line $\mathbb C\chi$ and $\supp\rho_A^{(1)}$ would have dimension at most $m-1$, contradicting $\supp\rho_A=\mathbb C\alpha_0\oplus\mathcal K_A$.  Hence Eq.~\eqref{eq:SM_A_support_update} holds. \hfill$\square$

\begin{lemma}[Entanglement passes to the residual]
\label{lem:SM_ent_passes}
If $\rho$ is entangled and Eq.~\eqref{eq:SM_peel_decomp} holds, then $\rho^{(1)}$ is entangled on its effective local supports.
\end{lemma}

\emph{Proof.}  If $\rho^{(1)}$ were separable, then Eq.~\eqref{eq:SM_peel_decomp} would express $\rho$ as a sum of separable positive operators: one product layer plus a separable residual.  Hence $\rho$ would be separable, a contradiction. \hfill$\square$

\begin{lemma}[Lifting active contacts]
\label{lem:SM_lifting}
If $\rho^{(1)}$ has an active pure contact in the direction $A\to B$, then the original state $\rho$ has an active pure contact in the same direction.
\end{lemma}

\emph{Proof.}  Let $x_0\in\mathcal K_A$ be a unit residual contact vector and let $x_1\in\mathcal K_A$, $x_1\perp x_0$, be a unit residual tangent vector.  For an operator $\omega$, write
\begin{equation}
    C_\omega(x,y)=(\bra{x}\otimes I)\omega(\ket{y}\otimes I).
\end{equation}
Define the linear map
\begin{equation}
    J:\mathcal K_A\to\mathbb C\alpha_0\oplus\mathcal K_A,
    \qquad
    Jx=x-\langle\gamma|x\rangle\alpha_0.
\end{equation}
Since $P_{\mathcal K_A}J=I_{\mathcal K_A}$, the map $J$ is injective.  Moreover, $Jx\perp\chi=\alpha_0+\gamma$.  The peeled product layer in Eq.~\eqref{eq:SM_peel_decomp} therefore contributes nothing to blocks between lifted vectors, and
\begin{equation}
    C_\rho(Jx,Jy)=C_{\rho^{(1)}}(x,y)
    \qquad (x,y\in\mathcal K_A).
    \label{eq:SM_lift_block_identity}
\end{equation}
Let
\begin{equation}
    A=C_{\rho^{(1)}}(x_0,x_0)=qP_{\varphi},
    \qquad
    B=C_{\rho^{(1)}}(x_0,x_1),
\end{equation}
where $q>0$ and $P_{\varphi}BP_{\varphi^\perp}\ne0$.  Set $u_0=Jx_0$, $u_1=Jx_1$, and
\begin{equation}
    \lambda=\frac{\langle u_0|u_1\rangle}{\|u_0\|^2},
    \qquad
    w_1=u_1-\lambda u_0.
\end{equation}
The injectivity of $J$ implies $w_1\ne0$.  Thus
\begin{equation}
    \widehat u_0=\frac{u_0}{\|u_0\|},
    \qquad
    \widehat u_1=\frac{w_1}{\|w_1\|}
\end{equation}
are a unit contact vector and an orthogonal unit tangent vector.  By Eq.~\eqref{eq:SM_lift_block_identity},
\begin{align}
    P_{\varphi} C_\rho(\widehat u_0,\widehat u_1)P_{\varphi^\perp}
    &=\frac{P_{\varphi}(B-\lambda A)P_{\varphi^\perp}}
            {\|u_0\|\,\|w_1\|} \\
    &=\frac{P_{\varphi}BP_{\varphi^\perp}}
            {\|u_0\|\,\|w_1\|}\ne0,
\end{align}
because $P_{\varphi} A P_{\varphi^\perp}=0$.  Hence normalization and tangent orthogonalization preserve the nonzero support--kernel component, and the lifted contact is active. \hfill$\square$

This lifting statement can be iterated.  The corresponding linear injections compose.  At each lift the unnormalized blocks are preserved, and the normalization and tangent-orthogonalization argument above preserves the nonzero support--kernel component.  An active contact of the final residual therefore lifts to an active contact of the original state.

\begin{lemma}[Pure entangled residuals]
\label{lem:SM_pure_terminal}
Every pure entangled residual has an active pure contact in both steering directions on its effective Schmidt supports.
\end{lemma}

\emph{Proof.}  Write a Schmidt decomposition
\begin{equation}
    \ket{\psi}=\sum_{j=1}^s s_j\ket{j_A}\ket{j_B},
    \qquad s\ge2,\quad s_j>0.
\end{equation}
For $A\to B$, the outcome $j=1$ prepares the pure Bob state $s_1^2P_{1_B}$.  Moving Alice in the tangent direction $\ket{2_A}$ produces the first-order block $s_1s_2|1_B\rangle\langle2_B|$, which connects the support line to the kernel.  The reverse steering direction is obtained by exchanging the parties. \hfill$\square$

\begin{corollary}[Rank-two degenerate branch]
\label{cor:SM_rank_two_branch}
Let $\rho$ have rank two and a degenerate pure contact in the direction $A\to B$.  Then, for unit representatives $u\in\mathbb C\alpha_0\oplus\mathcal K_A$ and $\zeta\in\mathcal K_A\otimes\mathcal H_B$,
\begin{equation}
    \rho=pP_{u\otimes\phi}+(1-p)P_{\zeta},
    \qquad 0<p<1,
    \label{eq:SM_rank_two_branch_decomp}
\end{equation}
where the Alice component of $u$ along $\alpha_0$ is nonzero.  If $\rho$ is entangled, then $\zeta$ is entangled and the same steering direction contains an active pure contact.
\end{corollary}

\emph{Proof.}  By Lemma~\ref{lem:SM_peel}, the degenerate contact gives
\begin{equation}
    \rho=a\ketbra{\chi}\otimes P_{\phi}+\rho^{(1)},
    \qquad \rho^{(1)}\ge0,
    \qquad \rank\rho^{(1)}=1.
\end{equation}
Normalizing the two positive summands gives Eq.~\eqref{eq:SM_rank_two_branch_decomp}, with $u$ the unit representative of the line $\mathbb C\chi$.  If $\zeta$ were product, Eq.~\eqref{eq:SM_rank_two_branch_decomp} would be a separable decomposition of $\rho$, so entanglement of $\rho$ forces $\zeta$ to be a pure entangled vector.

To identify the resulting active contact, write $u=c\alpha_0+u_\perp$ with $u_\perp\in\mathcal K_A$.  The Alice Schmidt support of the entangled vector $\zeta$ inside $\mathcal K_A$ has dimension at least two, hence is not contained in the subspace $\mathbb C u_\perp$; if $u_\perp=0$, this subspace is understood as $\{0\}$.  Choose a unit vector $\delta\in\mathcal K_A$ such that
\begin{equation}
    \langle\delta|u_\perp\rangle=0,
    \qquad
    b_0:=(\bra{\delta}\otimes I)\zeta\ne0.
\end{equation}
The outcome $\delta$ removes the product contribution at zeroth order in Eq.~\eqref{eq:SM_rank_two_branch_decomp} and prepares the pure Bob state $(1-p)\ketbra{b_0}$.  Since $\zeta$ has Schmidt rank at least two, there is a unit vector $\epsilon\in\mathcal K_A$, with $\epsilon\perp\delta$, such that
\begin{equation}
    b_1:=(\bra{\epsilon}\otimes I)\zeta
\end{equation}
is not collinear with $b_0$.  Along the projective path $(\delta+t\epsilon)/\sqrt{1+t^2}$ the conditional state is
\begin{equation}
    \frac{(1-p)\ketbra{b_0+t b_1}
          +p|\langle\delta+t\epsilon|u\rangle|^2P_{\phi}}{1+t^2}.
\end{equation}
The choice $\langle\delta|u_\perp\rangle=0$ makes the product contribution $O(t^2)$, while the component of $b_1$ orthogonal to $b_0$ produces a first-order support--kernel block.  Hence the new contact is active. \hfill$\square$

Thus the rank-two degenerate case is contained in the general Schur-complement recursion.

\section{S6. Finite rank-forcing recursion}
\label{sec:SM_recursion}

The preceding lemmas yield the directional low-rank theorem used in the Letter.

\begin{theorem}[Directional rank-forced projective steering]
\label{thm:SM_rank_forced_steering}
Let $\rho$ be an entangled state on effective local supports $\mathbb C^m\otimes\mathbb C^n$, $m,n\ge2$, and let $r=\rank\rho$.  If
\begin{equation}
    r\le1+\left\lfloor\frac{m-1}{n-1}\right\rfloor,
    \label{eq:SM_directional_rank_window}
\end{equation}
then $\rho$ is projectively steerable from $A$ to $B$.  If
\begin{equation}
    r\le1+\left\lfloor\frac{n-1}{m-1}\right\rfloor,
\end{equation}
then $\rho$ is projectively steerable from $B$ to $A$.
\end{theorem}

\emph{Proof.}  We prove the first statement.  The second is obtained by exchanging the parties.  If $r=1$, the state is pure entangled and Lemma~\ref{lem:SM_pure_terminal} gives an active contact; Proposition~\ref{prop:SM_sk_certificate} then gives projective steering.

Assume $r>1$.  By Lemma~\ref{lem:rank_forced_SM}, Eq.~\eqref{eq:SM_directional_rank_window} forces a pure contact.  If this contact is active, Proposition~\ref{prop:SM_sk_certificate} gives projective steering.  If it is degenerate, Lemmas~\ref{lem:SM_peel}, \ref{lem:SM_support_update}, and \ref{lem:SM_ent_passes} produce an entangled residual $\rho^{(1)}$ with
\begin{equation}
    r_1=r-1,
    \qquad m_1=m-1,
    \qquad n_1\le n.
\end{equation}
If the residual is not pure, its entanglement gives $m_1,n_1\ge2$.  The same rank-forcing inequality applies to the residual, since
\begin{align}
    m_1-1-(n_1-1)(r_1-1)
    &=(m-2)-(n_1-1)(r-2) \\
    &=[m-1-(n-1)(r-1)]
       +(n-2)+(n-n_1)(r-2)
       \ge0.
\end{align}
Hence, unless the residual is already pure, the residual again has a pure contact.  Repeating the active/degenerate dichotomy gives a finite process: each degenerate step lowers the rank by one, so after finitely many steps one reaches either an active contact or a pure entangled residual.  In the pure case, Lemma~\ref{lem:SM_pure_terminal} supplies an active contact.  Iterating Lemma~\ref{lem:SM_lifting} through all previous peels preserves the unnormalized blocks and the nonzero support--kernel component after normalization and tangent orthogonalization.  Proposition~\ref{prop:SM_sk_certificate} then proves steering of the original state. \hfill$\square$

\begin{corollary}[Rank-Two Projective Steering]
\label{cor:SM_rank_two}
Let $\rho$ be a rank-two entangled state on effective local supports $\mathbb C^m\otimes\mathbb C^n$.  If $m\ge n$, then $\rho$ is projectively steerable from $A$ to $B$.  If $n\ge m$, then $\rho$ is projectively steerable from $B$ to $A$.  In particular, if $m=n$, the state is two-way projectively steerable.
\end{corollary}

\emph{Proof.}  For $r=2$, Eq.~\eqref{eq:SM_directional_rank_window} is equivalent to $m\ge n$.  The exchanged condition is equivalent to $n\ge m$. \hfill$\square$

\section{S7. Consequences of the rank-forced window}
\label{sec:SM_scope}

For the fixed steering direction $A\to B$, the contraction space has dimension $m$, while the trusted rank-one locus sits inside matrices with $n$ rows.  This gives the condition
\begin{equation}
    r\le 1+\left\lfloor\frac{m-1}{n-1}\right\rfloor.
\end{equation}
After exchanging the parties, the corresponding condition is obtained by interchanging $m$ and $n$.  Thus, for a $d\otimes2$ state, measurements on the $d$-dimensional side force projective steering up to rank $r=d$.  The exchanged bound leaves no nontrivial mixed-rank window when $d>2$, while for $d=2$ the two directions coincide.

The rank-two case follows directly from the two inequalities.  If $m\ge n$, Theorem~\ref{thm:SM_rank_forced_steering} gives steering from $A$ to $B$; if $n\ge m$, the exchanged statement gives steering from $B$ to $A$.  Since one of these two orderings always holds, every rank-two entangled state is projectively steerable in the direction from the larger effective local dimension to the smaller one.  When $m=n$, both orderings hold and the conclusion is two-way projective steering.

The same rank window is compatible with the usual low-rank PPT comparison, but the two arguments play different roles.  The low-rank PPT separability theorem~\cite{HorodeckiLowRank2000} says that a PPT state of rank at most $\max\{m,n\}$ is separable.  In the directional window,
\begin{equation}
    r\le 1+\left\lfloor\frac{m-1}{n-1}\right\rfloor \le m\le\max\{m,n\},
\end{equation}
so an entangled state in this window must be NPT.  The support--kernel boundary certificate, together with the Schur-complement recursion for a degenerate contact, supplies the steering conclusion.  The negative partial-transpose minor is a parallel consequence of the same boundary block.

Beyond the forced rank window, the support--kernel certificate remains a direct test whenever a pure contact and a nonzero support--kernel tangent block are available by other means.  The test requires only the contact vector, a tangent direction on the untrusted side, and a kernel vector on the trusted side.  These data are often more accessible than an optimized steering inequality in exact algebraic families, low-rank tomography, variational ansatz states, tensor-network truncations, dissipative dark-state manifolds, and bipartite cuts of multipartite systems.  The rank-forcing lemma identifies when rank and dimension force a boundary contact, while the support--kernel certificate shows how the same block produces a negative partial-transpose minor and certifies projective steering.